\documentclass[letterpaper]{article} % DO NOT CHANGE THIS
\usepackage{aaai2026}  % DO NOT CHANGE THIS
\usepackage{times}  % DO NOT CHANGE THIS
\usepackage{helvet}  % DO NOT CHANGE THIS
\usepackage{courier}  % DO NOT CHANGE THIS
\usepackage[hyphens]{url}  % DO NOT CHANGE THIS
\usepackage{graphicx} % DO NOT CHANGE THIS
\usepackage{natbib}  % DO NOT CHANGE THIS AND DO NOT ADD ANY OPTIONS TO IT
\usepackage{caption} % DO NOT CHANGE THIS AND DO NOT ADD ANY OPTIONS TO IT
\usepackage{algorithmic}
\usepackage{algorithm}

\usepackage{amsthm,amsmath}
\usepackage{tikz}
\usepackage{newfloat}
\usepackage{listings}
\DeclareCaptionStyle{ruled}{labelfont=normalfont,labelsep=colon,strut=off} % DO NOT CHANGE THIS
\floatstyle{ruled}
\newfloat{listing}{tb}{lst}{}
\floatname{listing}{Listing}
\newtheorem{theorem}{Theorem}
\newtheorem{lemma}[theorem]{Lemma}
\newtheorem{claim}[theorem]{Claim}
\newtheorem{corollary}[theorem]{Corollary}
\newtheorem{definition}[theorem]{Definition}

\DeclareMathOperator*{\argmax}{arg\,max}
\DeclareMathOperator*{\im}{improve}

\allowdisplaybreaks

\title{Computing Approximately Proportional Allocations of  Indivisible Goods:\\Beyond Additive and Monotone Valuations}
\author{Martin Jupakkal Andersen, Ioannis Caragiannis, Anders Bo Ipsen, and Alexander S{\o}ltoft}
\affiliations{Department of Computer Science, Aarhus University, {\AA}bogade 34, 8200 Aarhus N, Denmark}

\begin{document}

\maketitle

\begin{abstract}
Although approximate notions of envy-freeness—such as envy-freeness up to one good (EF1)—have been extensively studied for indivisible goods, the seemingly simpler fairness concept of proportionality up to one good (PROP1) has received far less attention. For additive valuations, every EF1 allocation is PROP1, and well-known algorithms such as Round-Robin and Envy-Cycle Elimination compute such allocations in polynomial time. PROP1 is also compatible with Pareto efficiency, as maximum Nash welfare allocations are EF1 and hence PROP1.

We ask whether these favorable properties extend to non-additive valuations. We study a broad class of allocation instances with {\em satiating goods}, where agents have non-negative valuation functions that need not be monotone, allowing for negative marginal values. We present the following results:

\begin{itemize}
\item EF1 implies PROP1 for submodular valuations over satiating goods, ensuring existence and efficient computation via Envy-Cycle Elimination for monotone submodular  valuations;

\item Round-robin computes a partial PROP1 allocation after the second-to-last round for satiating submodular goods and a complete PROP1 for monotone submodular valuations;

\item PROP1 allocations for satiating subadditive goods can be computed in polynomial-time;

\item Maximum Nash welfare allocations are PROP1 for monotone submodular goods, revealing yet another facet of their ``unreasonable fairness.''
\end{itemize}
\end{abstract}


\section{Introduction}

Proportionality~\cite{S48} is the most important shared-based fairness concept in the fair division literature.  An allocation is proportional if each of the $n$ agents gets at least a $1/n$-th of all items in terms of value. Unfortunately, for allocation problems with indivisible items, proportional allocations may not exist. The relaxed concept of proportionality up to some item (PROP1), introduced by~\citet{CFS17}, comes to address this limitation. PROP1 is well understood in allocation instances with goods and agents with additive valuations. However, our understanding of it in settings with more general agent valuations is very limited. We aim to fill this gap in this work.

PROP1 relaxes proportionality in a similar in spirit way in which EF1 (envy-freeness up to some item; introduced by~\citet{B11}) relaxes envy-freeness. As proved by \citet{CFS17}, for additive valuations, EF1 implies PROP1; thus, several results that guarantee the existence and efficient computation of EF1 allocations trivially imply the existence and efficient computation of PROP1 allocations as well.  For example, the well-known Envy-Cycle Elimination algorithm of~\citet{LMMS04} and the folklore Round-Robin algorithm produce EF1 allocations while the maximum Nash welfare allocations are EF1 and Pareto-optimal~\cite{CKM+19}. As a corollary of the EF1-to-PROP1 implication for additive valuations, we can replace EF1 with PROP1 in these three results.

Furthermore, the celebrated Envy-Cycle Elimination algorithm works on allocation instances with monotone goods and computes an EF1 allocation. So, one would even hope to use it to get PROP1 allocations on instances with monotone valuations. Unfortunately, such a positive result is not possible. For example, consider the instance with two agents and three items such that both agents have valuation $2$ for the whole set of items and value $0$ for every strict subset of it. The valuation function is clearly monotone. However, in any allocation, some agent will get only one item and her value will be $0$ even after adding one additional item to her bundle, i.e., below the proportionality threshold. 

The above example implies that for sufficiently general (e.g., {\em super-additive}) valuations, EF1 allocations are not PROP1. What is the broadest class of allocation instances in which EF1 implies PROP1? Can we then use well-known algorithms for EF1 (such as the Envy-Cycle Elimination and Round-Robin) to produce PROP1 allocations? What is the broadest class of allocation instances in which a PROP1 allocation can be computed in polynomial time? What is the broadest class of allocation instances in which PROP1 is compatible with Pareto-optimality? These are the questions we study in this paper, making substantial progress towards understanding PROP1.

\paragraph{Our contribution.} We address the questions above and present a list of new results on PROP1 allocations for instances with more general than additive valuation functions. At the conceptual level, we consider valuation functions over {\em satiating goods}. Such functions return non-negative values for bundles but the marginal value of adding a good to a bundle can be negative. An agent typically has positive value for getting a single item but adding this item into a bundle that already contains some other items can decrease her value. The definitions of submodular and subadditive valuation functions are naturally extended to satiating goods. Submodularity means non-increasing marginal value while subadditivity means that the valuation for a bundle of items is not higher than the sum of valuations for any two disjoint subsets of it. 

Our first technical result is that EF1 allocations are also PROP1 in allocation instances with satiating submodular goods. The implication is also true for allocation instances with two agents and satiating subadditive valuations but is not true for allocation instances with more agents and slightly more general than monotone submodular valuations. As a corollary, we get that the Envy-Cycle Elimination algorithm computes PROP1 allocations for instances with monotone submodular goods. These results appear in Section~\ref{sec:ef1-implies-prop1}.

Next, in Section~\ref{sec:round-robin}, we study the Round-Robin algorithm. Round-Robin is well-known to produce EF1 allocation for additive goods but fails to do so for more general valuation functions; see~\citet{ABLLR23}. Hence, the implication betwen EF1 and PROP1 does not have any implication for Round-Robin in non-additive instances. Somewhat surprisingly, we prove that Round-Robin does produce a PROP1 allocation on instances with monotone submodular goods. Actually, it almost does so for satiating submodular goods as well. The partial allocation it computes after its second-to-last round is PROP1. These results are best possible. The final allocation of Round-Robin may not be PROP1 when applied on allocation instances with satiating submodular goods or with monotone valuations that are slightly more general than submodular.

Our strongest algorithmic result is a new algorithm that computes a PROP1 allocation for satiating subadditive goods. The algorithm starts with an arbitrary allocation and gradually satisfies the PROP1 conditions for more and more agents by repeatedly moving items from bundle to bundle. The proof of correctness uses a nice potential function argument. The algorithm can be combined with Round-Robin to give a considerably faster algorithm for satiating submodular goods. These results are presented in Section~\ref{sec:subadditive}.

Finally, we address the question of whether PROP1 is compatible with Pareto-optimality. Our main positive result is another facet of the unreasonable fairness of maximum Nash welfare allocations. Such allocations are PROP1 and Pareto-optimal on instances with monotone submodular goods. This is the broadest class of monotone instances in which PROP1 and Pareto-optimality are always compatible. We present instances with slightly more general monotone valuation functions than submodular where no PROP1 allocation is Pareto-optimal. These results appear in Section~\ref{sec:prop1-and-pareto}.

We continue by discussing the related literature in the rest of this section. We present useful preliminary definitions in Section~\ref{sec:prelim} that are necessary for the presentation of our technical results in Sections~\ref{sec:ef1-implies-prop1}-\ref{sec:prop1-and-pareto}, and conclude with open problems in Section~\ref{sec:open}.

\paragraph{Related work.}
PROP1 was introduced by~\citet{CFS17}. They considered additive valuation functions and observed, among other results, that EF1 allocations are PROP1. Three papers  by~\citet{AMS20}, \citet{BK19}, and \citet{MG20} consider the compatibility of PROP1 and Pareto-optimality and present related algorithmic results. For instances with additive goods, PROP1 seems to be the simplest fairness property; the recent work by \citet{GS25} summarize how PROP1 is implied by other fairness properties (see also the related software at \url{https://sharmaeklavya2.github.io/cpigjs/fairDiv/}). 

Non-additive valuation functions have been considered extensively recently. Submodular valuation functions are central in the study of allocation problems since the seminal work of~\citet{LLN06}, who also introduced the class XOS (or fractionally subadditive valuations). In fair division, \citet{ABLLR23} and~\citet{BK20} have studied the Round-Robin protocol on allocation instances with submodular valuations. Subadditive valuations have been considered in relation to approximate versions of the fairness notions EFX~\cite{CGM21,FMP24,BS24} and MMS~\cite{SS25}; see also the survey by \citet{AABF+23}. Subadditive valuations have been important in the study of utilitarian and Nash welfare maximization by, e.g., \citet{F09} and \citet{DLRV24}, respectively.

The above results refer to goods, i.e., monotone valuation functions with non-negative marginals. In fair division, non-monotonicity has been studied in the model of combined goods and chores of~\citet{ACIW22}. These valuation functions are additive but an item can have positive value for an agent and negative to another. Our definition of satiating valuations has non-negative valuations but allows for negative marginals. Similar valuations (with possibly negative values for non-empty bundles that do not contain all items) have been considered very recently by~\citet{BV25}, in relation to the equitability fairness concept. The literature on optimization of set functions has focused extensively on satiating submodular valuations, starting with the work of~\citet{FMV11}; see also the related survey by~\citet{KG14}.

\section{Preliminaries}\label{sec:prelim}
We consider {\em allocation instances} in which a set $M$ of $m$ {\em goods} (or items) has to be allocated to $n$ {\em agents}. We identify the agents using the positive integers in $[n]=\{1,2, ..., n\}$. Each agent has a {\em valuation function} which returns the value the agent has for each set (or {\em bundle}) of goods. We denote by $v_i$ the valuation function of agent $i\in [n]$; it is normalized with $v_i(\emptyset)=0$ and takes non-negative values, i.e., $v_i(S)\geq 0$, for every bundle $S\subseteq M$ of goods. When $S$ is a singleton with, say, $S=\{o\}$, we simplify notation and use $v_i(o)$ instead of $v_i(\{o\})$. The valuation function $v_i$ is called {\em additive} if $v_i(S)=\sum_{o\in S}{v_i(o)}$ for every set of goods $S\subseteq M$.

For two disjoint bundles of goods $S$ and $T$, we use the notation $v_i(T|S):=v_i(T\cup S)-v_i(S)$ to denote the {\em marginal value} the bundle $T$ has for agent $i\in [n]$ when it is added to the bundle $S$. When $T$ is a singleton with, say, $T=\{o\}$, we simplify notation to $v_i(o|S)$ (instead of $v_i(\{o\}|S)$). 

The valuation function $v_i$ is called {\em submodular} if $v_i(S\cup T)+v_i(S\cap T)\leq v_i(S)+v_i(T)$ for every sets $S,T\subseteq M$. A submodular valuation function satisfies $v_i(o|S)\geq v_i(o|T)$ for every two sets of goods $S$ and $T$ with $S\subseteq T\subseteq M$ and good $o\not\in T$.

Our definition for valuation functions allows for marginal values to have any sign. We use the term {\em monotone submodular goods} for submodular valuation functions with non-negative marginal values and {\em satiating submodular goods} for general submodular valuation functions.

The next claim applies to the most general definition of satiating submodular goods and gives an equivalent definition of submodularity, which we use in our proofs.

\begin{claim}\label{claim:submodularity-alt}
    The submodular valuation function $v_i$ satisfies $v_i(X|S)\geq v_i(X|T)$ for every set of goods $X\subseteq M$ and every two sets of goods $S$ and $T$ with $S\subseteq T \subseteq M$ and $X\cap T=\emptyset$.
\end{claim}

Very often in our proofs, we use the following property.
\begin{claim}\label{claim:submodular-vs-additive}
The submodular valuation function $v_i$ satisfies $v_i(X\cup T|S)\leq v_i(X|S)+v_i(T|S)$ for every mutually disjoint sets of goods $X$, $S$, and $T$. Hence, it also satisfies
$v_i(T|S) \leq \sum_{o\in T}{v_i(o|S)}$ for every two disjoint sets of goods $S$ and $T$.
\end{claim}

\begin{proof}
    To prove the first part of Claim~\ref{claim:submodular-vs-additive}, we use the definition of marginal value and submodularity (Claim~\ref{claim:submodularity-alt}) to get $v_i(X\cup T|S) =v_i(X|S)+v_i(T|X\cup S)\leq v_i(X|S)+v_i(T|S)$.
\end{proof}

The valuation function $v_i$ is called {\em subadditive} if $v_i(S\cup T)\leq v_i(S)+v_i(T)$ for every two sets of goods $S,T\subseteq M$. Again, we use the terms {\em monotone} and {\em satiating subadditive goods} to distinguish between subadditive valuation functions with non-negative marginal and general subadditive valuation functions, respectively.

Every submodular valuation function is also subadditive. In some proofs, we use (monotone) {\em XOS} valuation functions defined as follows. A valuation function $v_i(\cdot)$ over bundles of items is XOS if there are additive valuation functions $f_1(\cdot)$, $f_2(\cdot)$, ..., $f_k(\cdot)$ such that $v_i(S)=\max_{j\in [k]}{f_j(S)}$. A monotone submodular valuation function is also XOS and a XOS valuation function is also monotone subadditive. Figure~\ref{fig:valuations} summarizes the different valuation functions used in the paper.

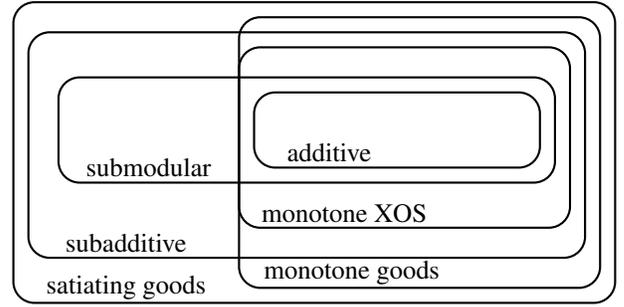
\begin{figure}[ht]
\centering
\begin{tikzpicture}
    \draw[rounded corners=8pt, thick] (0,0) rectangle (8,4);
    \node at (1.5,0.2) {satiating goods};
    \draw[rounded corners=8pt, thick] (3,0.2) rectangle (7.8,3.8);
    \node at (4.5,0.4) {monotone goods};
    \draw[rounded corners=8pt, thick] (0.2,0.6) rectangle (7.6,3.6);
    \node at (1.5,0.8) {subadditive};
    \draw[rounded corners=8pt, thick] (3,1) rectangle (7.4,3.4);
    \node at (4.4,1.2) {monotone XOS};
    \draw[rounded corners=8pt, thick] (0.6,1.6) rectangle (7.2,3);
    \node at (1.8,1.8) {submodular};
    \draw[rounded corners=8pt, thick] (3.2,1.8) rectangle (7,2.8);
    \node at (4.2,2) {additive};
\end{tikzpicture}
\caption{Relation between the valuation functions over goods used in the paper.}
\label{fig:valuations}
\end{figure}

An allocation $A=(A_1, A_2, ..., A_n)$ is a disjoint partition of the set of goods $M$ to the $n$ agents. For $i\in [n]$, the set $A_i$ indicates the bundle of goods allocated to agent $i$. The allocation $A$ is {\em envy-free} if 
$v_i(A_i)\geq v_i(A_j)$ for every pair of agents $i,j\in [n]$ (i.e., agent $i$ weakly prefers her own bundle to the bundle of agent $j$) and is {\em proportional} if $v_i(A_i)\geq \frac{1}{n}\cdot v_i(M)$ (i.e., agent $i$ has a value for her bundle that is at least as high as her {\em proportionality threshold}). An allocation is {\em envy-free up to some good} (EF1) if for every two agents $i,j\in [n]$, it is either $v_i(A_i)\geq v_i(A_j)$ or there exists a good $g$ in the bundle of agent $j$ such that $v_i(A_i)\geq v_i(A_j\setminus \{g\})$. An allocation is {\em proportional up to some good} (PROP1) if for every agent $i\in [n]$, it is either $v_i(A_i)\geq \frac{1}{n}\cdot v_i(M)$, or there exists a good $g$ not allocated to agent $i$ such that $v_i(A_i\cup\{g\})\geq \frac{1}{n}\cdot v_i(M)$. These definitions apply to the most general case of satiating goods. For monotone goods, the first condition in the definition of EF1 and PROP1 is redundant.

\section{Does EF1 imply PROP1?}\label{sec:ef1-implies-prop1}
Our first technical result strengthens the well-known fact (see~\citealt[Lemma~3.6]{CFS17}) that EF1 allocations are also PROP1 to allocation instances with non-additive and non-monotone valuations. 
\begin{theorem}\label{thm:ef1-implies-prop1}
    In any allocation instance with satiating submodular goods, an EF1 allocation is also PROP1.
\end{theorem}

\begin{proof}
Consider an allocation instance with $n$ agents having submodular valuations for subsets of a set of satiating goods $M$, and let $A=(A_1, A_2, ..., A_n)$ be an EF1 allocation.

Let $i\in [n]$ be an agent. Define $\widetilde{N}:=\{j\in [n]:v_i(A_i)\geq v_i(A_j)\}$ and observe that $\widetilde{N}$ is non-empty since $i\in \widetilde{N}$. If $\widetilde{N}=[n]$, then applying the condition $v_i(A_i)\geq v_i(A_j)$ for $j\in [n]$, we get 
\begin{align}\label{eq:ef1-implies-prop1-easy}
    n\cdot v_i(A_i) &\geq \sum_{j\in [n]}{v_i(A_j)} \geq v_i(M),
\end{align}
completing the proof. The second inequality in Equation (\ref{eq:ef1-implies-prop1-easy}) follows by the subadditivity of the valuation function $v_i$.

Now, assume that $\widetilde{N}\not=[n]$. For each agent $j\in [n]\setminus\widetilde{N}$, let $o_j$ be an (arbitrary) item in the bundle $A_j$ such that $v_i(A_i)\geq v_i(A_i\setminus \{o_j\})$; since allocation $A$ is EF1 and $v_i(A_i)<v_i(A_j)$, such a good clearly exists. Let $O:=\{o_j:j\in [n]\setminus \widetilde{N}\}$ and set $o^*:=\argmax_{o\in O}{v_i(o|A_i)}$. We will prove the theorem by showing that $v_i(M)\leq n\cdot \max\{v_i(A_i),v_i(A_i\cup \{o^*\})\}$. 
    
Notice that the set $M\setminus O$ is the disjoint union of the bundles $A_j$ for $j\in \widetilde{N}$ and $A_j\setminus\{o_j\}$ for $j\in [n]\setminus \widetilde{N}$. By the subadditivity of the valuation function $v_i$, we have
\begin{align}\nonumber
    v_i(M\setminus O) &\leq \sum_{j\in \widetilde{N}}{v_i(A_j)} +\sum_{j\in [n]\setminus \widetilde{N}}{v_i(A_j\setminus \{o_j\})}\\\label{eq:m-setminus-o}
    &\leq  n\cdot v_i(A_i).
    \end{align}
    The last inequality follows since the allocation $A$ is EF1.

By the definition of set $O$, we have $A_i\cap O=\emptyset$. Thus, $A_i\subseteq M\setminus O$ and
    \begin{align}\nonumber
        v_i(O|M\setminus O) &\leq v_i(O|A_i)
        \leq \sum_{o\in O}{v_i(o|A_i)}\\\nonumber
        &\leq \sum_{o\in O}{v_i(o^*|A_i)}\\\label{eq:o-given-m-setminus-o}
        &=(n-|\widetilde{N}|)\cdot v_i(o^*|A_i).
    \end{align}
    The first inequality follows by the submodularity of valuation function $v_i$ (Claim~\ref{claim:submodularity-alt}), the second one by Claim~\ref{claim:submodular-vs-additive}, and the third one by the definition of $o^*$.

By Equations~(\ref{eq:m-setminus-o}) and (\ref{eq:o-given-m-setminus-o}), we have
\begin{align}\nonumber
    v_i(M) &= v_i(M\setminus O) + v_i(O|M\setminus O)\\\label{eq:ef1-implies-prop1-sum-of-two-terms}
    &\leq n\cdot v_i(A_i) + (n-|\widetilde{N}|) \cdot v_i(o^*|A_i).
\end{align}
Now, if $v_i(o^*|A_i)<0$, Equation (\ref{eq:ef1-implies-prop1-sum-of-two-terms}) immediately implies $v_i(M)\leq n\cdot v_i(A_i)$. Otherwise, we get
\begin{align*}
    v_i(M) &\leq n\cdot v_i(A_i)+(n-|\widetilde{N})\cdot v_i(o^*|A_i)\\
    &\leq n\cdot v_i(A_i)+n\cdot v_i(o^*|A_i)\\
    &= n\cdot v_i(A_i\cup\{o^*\}),
\end{align*}
again completing the proof.
\end{proof}

Theorem~\ref{thm:ef1-implies-prop1} has important algorithmic implications. Using the well-known result of~\citet{LMMS04} that the Envy-Cycle Elimination algorithm produces EF1 allocations when applied to monotone allocation instances, we get the following corollary.
\begin{corollary}\label{cor:ece}
    Given an allocation instance with monotone submodular goods, the Envy-Cycle Elimination algorithm returns a PROP1 allocation.
\end{corollary}
We will extend and significantly improve Corollary~\ref{cor:ece} in Sections~\ref{sec:round-robin} and~\ref{sec:subadditive}, respectively.

The result in Theorem~\ref{thm:ef1-implies-prop1} is best possible. The proof of the next theorem uses an allocation instance with monotone valuations that are slightly more general than submodular.
\begin{theorem}\label{thm:ef1-not-imply-prop1}
    There exists an allocation instance with monotone XOS goods that has an EF1 allocation that is not PROP1.
\end{theorem}

\begin{proof}
    Consider an allocation instance with three agents and seven items. Each agent has the valuation function $v$ depicted in Table~\ref{tab:ef1-not-imply-prop1}. Notice that $v(S)$ depends only on the cardinality of $S$. We can verify that $v$ is XOS. Indeed, using the positive integers in $[7]$ to identify the items, we have $v(S)=\max_{j=0,1, .., 7}{\sum_{g\in S}{f_j(g)}}$ with $f_0(g)=1$ for $g\in [7]$,  $f_g(g)=2$ for $g\in [7]$, and $f_t(g)=0$ for $t,g\in [7]$ with $t\not=g$.
\begin{table}[ht]
    \centering
    \begin{tabular}{c|c c c c c c c}
        $|S|$ & $1$ & $2$ & $3$ & $4$ & $5$ & $6$ & $7$ \\\hline
        $v(S)$ & $2$ & $2$ & $3$ & $4$ & $5$ & $6$ & $7$
    \end{tabular}
    \caption{The valuation function of the agents in the instance used in the proof of Theorem~\ref{thm:ef1-not-imply-prop1}.}
    \label{tab:ef1-not-imply-prop1}
\end{table}

Now, consider an allocation $A=(A_1, A_2, A_3)$ with $|A_1|=1$, $|A_2|=|A_3|=3$. Clearly, agents $2$ and $3$ are non-envious. For every strict subset of bundles $A_2$ and $A_3$, agent $1$ has value $2$, i.e., equal to her value $v_1(A_1)$ in allocation $A$. Hence, allocation $A$ is EF1. It is not PROP1, though. Agent $1$ has still value $2$ for any bundle obtained by adding an extra item to her bundle $A_1$, while her proportionality threshold is $7/3$.
\end{proof}

We remark that the proof of Theorem~\ref{thm:ef1-not-imply-prop1} uses an instance with three agents. Interestingly, for two-agent allocation instances and satiating subadditive goods, EF1 implies PROP1 as the next statement indicates.
\begin{theorem}\label{thm:ef1-implies-prop1-two-agents}
In allocation instances with satiating subadditive goods and two agents, any EF1 allocation is also PROP1.
\end{theorem}

\begin{proof}
Consider an allocation instance with two agents having subadditive valuations for subsets of a set of satiating goods and let $A=(A_1,A_2)$ be an EF1 allocation. Without loss of generality, we will focus on agent $1$ and prove that allocation $A$ satisfies the PROP1 condition for this agent.

First, assume that $v_1(A_1)\geq v_1(A_2)$. Using subadditivity and this inequality, we get
\begin{align*}
    v_1(A_1\cup A_2) &\leq v_1(A_1)+v_1(A_2) \leq 2\cdot v_1(A_1),
\end{align*}
completing the proof. Otherwise, due to the fact that allocation $A$ is EF1, there exists an item $g$ in the bundle $A_2$, such that $v_1(A_1)\geq v_1(A_2\setminus\{g\})$. Using subadditivity and this inequality, we get
\begin{align*}
    v_1(A_1\cup A_2) &\leq v_1(A_2\setminus\{g\})+v_1(A_1\cup \{g\})\\
    &\leq v_1(A_1)+v_1(A_1\cup \{g\})\\
    &\leq 2\cdot \max\{v_1(A_1),v_1(A_1\cup \{g\})\},
\end{align*}
again completing the proof.
\end{proof}

\section{Does Round-Robin compute PROP1 allocations?}\label{sec:round-robin}
We now focus on the Round-Robin protocol, arguably the simplest algorithm for allocating indivisible goods. Round-Robin starts with an empty allocation and runs in rounds. In each round, the agents act in a fixed order. When it is the turn of an agent to act, they pick the best item that is still available, i.e., the available item that has the maximum marginal value for the agent. For satiating goods, this marginal value can be negative.

Our main result in this section (Corollary~\ref{cor:rr}) applies to allocation instances with monotone submodular goods. For this case, \citet{ABLLR23} have proven that the allocation returned by Round-Robin is not always EF1. Somewhat surprisingly, we prove that it is PROP1, using different arguments than those developed in Section~\ref{sec:ef1-implies-prop1}.

For satiating submodular goods, Round-Robin achieves PROP1 only partially, as the next statement indicates.
\begin{theorem}\label{thm:rr-second-last-phase}
    In every allocation instance with satiating submodular goods, the partial allocation produced by Round-Robin after the second-to-last round is PROP1.
\end{theorem}

\begin{proof}
Consider an allocation instance with $n$ agents having submodular valuations for subsets of a set of satiating goods $M$. 

We will assume that $n$ divides $|M|$. If this was not the case, we can transform the initial instance to a new one by adding $\left\lceil|M|/n\right\rceil\cdot n-|M|$ {\em dummy} items and extend the valuation functions as follows: For every dummy item $g_1$, agent $i\in [n]$, and bundle of items $S$ such that $g_1\not\in S$, it holds $v_i(g_1|S)=0$ and, furthermore, for every other item $g_2$ with $g_2\not=g_1$ and $g_2\not\in S$, it holds $v_i(g_2|S)=v_i(g_2|S\cup\{g_1\})$. It can be easily seen that any bundle of items in the new instance has the same value for any agent in the original instance after removing the dummy items. Hence, any PROP1 allocation for the modified instance yields a PROP1 instance for the original instance.

Let $i\in [n]$ be an agent, $A=(A_1, ..., A_n)$ be the allocation computed by Round-Robin, and let $L=|M|/n$ be the number of rounds, i.e., $|A_i|=L$. For $k\in [L]$, denote by $A_i^k$ the set of the $k$ first items allocated by Round-Robin to agent $i$. Let $o^*$ be an item that satisfies $o^*\in \argmax_{o\in M\setminus A_i^{L-1}}{v_i(o|A_i^{L-1})}$. We will prove the theorem by showing that $v_i(M)\leq n\cdot v_i(A_i^{L-1}\cup \{o^*\})$.

We define the following sets of items allocated to all agents besides agent $i$. For $k=1, ..., L$, the set $\Pi_k$ contains all the items allocated to agents different from $i$ before the $k$-th item is allocated to agent $i$. Also, we define as $\Pi_{L+1}$ the set of items allocated to agents different from $i$ after the last item is allocated to agent $i$. Notice that $|\Pi_k\setminus \Pi_{k-1}|=n-1$ for $k=2, ..., L$, and $|\Pi_1\cup \Pi_{L+1}|= n-1$. Using these definitions, we have 
\begin{align}\nonumber
    v_i(M)&=v_i(A_i^{L-1})\\\nonumber
    &\quad\, +v_i\left(\Pi_1\cup \Pi_{L+1}\cup(A_i\setminus A_i^{L-1})|A_i^{L-1}\right)\\\label{eq:three-terms}
    &\quad\,  +\sum_{k=2}^L{v_i\left(\Pi_k\setminus \Pi_{k-1}|\Pi_{k-1}\cup \Pi_{L+1} \cup A_i\right)}.
\end{align}
We will bound the second and third terms in the RHS of equation (\ref{eq:three-terms}) separately.

First, using Claim~\ref{claim:submodular-vs-additive}, the definition of item $o^*$, and the fact that the set $\Pi_1\cup \Pi_{L+1}\cup (A_i\setminus A_i^{L-1})$ consists of $n$ items which also belong to $M\setminus A_i^{L-1}$, we have
\begin{align}\nonumber
    &v_i(\Pi_1\cup \Pi_{L+1}\cup (A_i\setminus A_i^{L-1})|A_i^{L-1})\\\nonumber
    &\leq \sum_{o\in \Pi_1\cup \Pi_{L+1}\cup (A_i\setminus A_i^{L-1})}{v_i(o|A_i^{L-1})}\\\nonumber
    &\leq \sum_{o\in \Pi_1\cup \Pi_{L+1}\cup (A_i\setminus A_i^{L-1})}{v_i(o^*|A_i^{L-1})}\\\label{eq:second-term}
    &= n\cdot v_i(o^*|A_i^{L-1}).
\end{align}

Next, recall that the $j$-th item that is allocated by Round-Robin to agent $i$ with $j\in [L-1]$ is at least as valuable in terms of marginal value to agent $i$ as any item that is allocated by Round-Robin later. We formally state this observation as follows.

\begin{claim}\label{claim:rr}
    For every $j\in [L]$ and every good $o\in M\setminus (\Pi_j\cup A_i^j)$, it holds that $v_i(A_i^j\setminus A_i^{j-1}|A_i^{j-1})\geq v_i(o|A_i^{j-1})$.
\end{claim}

We are now ready to bound the third term in the RHS of Equation (\ref{eq:three-terms}) above. We have
\begin{align}\nonumber
    &\sum_{k=2}^L{v_i\left(\Pi_k\setminus \Pi_{k-1}|\Pi_{k-1}\cup \Pi_{L+1}\cup A_i\right)}\\\nonumber
    &\leq \sum_{k=2}^L{\sum_{o\in \Pi_k\setminus \Pi_{k-1}}{v_i\left(o|\Pi_{k-1}\cup \Pi_{L+1}\cup A_i\right)}}\\\nonumber
    &\leq \sum_{k=2}^L{\sum_{o\in \Pi_k\setminus \Pi_{k-1}}{v_i\left(o|A_i^{k-2}\right)}}\\\nonumber
    &\leq \sum_{k=2}^L{\sum_{o\in \Pi_k\setminus \Pi_{k-1}}{v_i\left(A_i^{k-1}\setminus A_i^{k-2}|A_i^{k-2}\right)}}\\\nonumber
    &=(n-1)\cdot \sum_{k=2}^L{v_i\left(A_i^{k-1}\setminus A_i^{k-2}|A_i^{k-2}\right)}\\\label{eq:third-term}
    &=(n-1)\cdot v_i(A_i^{L-1}).
\end{align}
The first inequality follows by Claim~\ref{claim:submodular-vs-additive}, the second by the definition of submodularity and since $A_i^{k-2}\subseteq \Pi_{k-1}\cup \Pi_{L+1}\cup A_i$, and the third by applying Claim~\ref{claim:rr} with $j=k-1$ and since $\Pi_k\setminus \Pi_{k-1}\subseteq M\setminus (\Pi_{k-1}\cup A_i^{k-1})$. The first equality follows since $\Pi_k\setminus \Pi_{k-1}$ consists of $n-1$ terms and the second by definition.

Now, using the inequalities (\ref{eq:second-term}) and (\ref{eq:third-term}), Equation (\ref{eq:three-terms}) yields
\begin{align*}
    v_i(M) &\leq n\cdot v_i(A_i^{L-1})+n\cdot v_i(o^*|A_i^{L-1})\\
    &=n\cdot v_i(A_i^{L-1}\cup\{o^*\}),
\end{align*}
as desired.
\end{proof}

One might think that the fact that Theorem~\ref{thm:rr-not-prop1-non-mon} refers to the partial allocation computed at the end of the second-to-last round, as opposed to the final complete allocation returned by Round-Robin, is just a weakness in our analysis. This is not the case, though, as the next statement shows. 
\begin{theorem}\label{thm:rr-not-prop1-non-mon}
    There exists an allocation instance with satiating submodular goods for which Round-Robin does not produce a PROP1 allocation.
\end{theorem}

\begin{proof}
We use an allocation instance with two agents and six items $a$, $b$, $c$, $d$, $g_1$, and $g_2$. Agent $1$ has an additive valuation function, with value $1$ for items $a$, $b$, and $c$ and value $0$ for items $d$, $g_1$, and $g_2$. Agent $2$ has the valuation function defined in Table~\ref{tab:satiating}. Items $g_1$, $g_2$, $a$, $b$, and $c$ have constant marginal values $v(\cdot|S)$ of $3$, $3$, $2$, $2$, and $2$, respectively, when added to a bundle $S$ not containing item $d$. Item $d$ has marginal value $v_2(d|S)$ equal to $-5$ if set $S$ contains both $g_1$ and $g_2$, and equal to $2$ otherwise. As marginal values are non-increasing but possibly negative, this is an instance of satiating submodular goods.

\begin{table}[ht]
    \centering
    \begin{tabular}{c c|c }
        item $o$ & bundle $S$ & $v_2(o|S)$ \\\hline
        $g_1$, $g_2$ & $d\not\in S$ & $3$\\
        $a,b,c$ & $d\not\in S$ & $2$\\
        $d$ & $\{g_1,g_2\}\not\subseteq S$ & $2$\\
        $d$ & $\{g_1,g_2\}\subseteq S$ & $-5$\\
    \end{tabular}
    \caption{The valuation function of agent $2$ in the proof of Theorem~\ref{thm:rr-not-prop1-non-mon}.}
    \label{tab:satiating}
\end{table}

By applying Round-Robin, agent $1$ picks the items $a$, $b$, and $c$ in the first three rounds. Agent $2$ picks items $g_1$ and $g_2$ in the first two round and is left with item $d$ in the third round. We have $v_2(M)=7$ while $v_2(\{g_1,g_2,d,o\})=3<v_2(M)/2$ for every $o\in \{a,b,c\}$, violating the PROP1 condition.
\end{proof}

Our next positive result for monotone submodular goods follows easily using Theorem~\ref{thm:rr-second-last-phase} and by arguing about the additional value the agents get in the last round of Round-Robin.

\begin{corollary}\label{cor:rr}
    For every allocation instance with monotone submodular goods, Round-Robin returns a PROP1 allocation.
\end{corollary}

\begin{proof}
Consider an allocation instance with $n$ agents having submodular valuations for subsets of a set of satiating goods $M$. Let $L=\left\lceil |M|/n \right\rceil$ be the number of rounds executed by Round-Robin. For agent $i\in [n]$ and $k\in [L]$, denote by $A_i^k$ the set of items allocated by Round-Robin to agent $i$ in rounds $1, 2, ..., k$ and let $A=(A_1, A_2, ..., A_n)$ be the final complete allocation. By Theorem~\ref{thm:rr-second-last-phase}, we get that either $v_i(M)\leq n\cdot v_i(A_i^{L-1})$ or there exists an item $g\not\in A_i^{L-1}$ such that $v_i(M)\leq n\cdot v_i(A_i^{L-1}\cup \{g\})$. 

If $v_i(M)\leq n\cdot v_i(A_i^{L-1})$, monotonicity yields $v_i(A_i^{L-1})\leq v_i(A_i)$ and, hence, $v_i(M)\leq n\cdot v_i(A_i)$ as well. If $v_i(M)\leq n\cdot v_i(A_i^{L-1}\cup \{g\})$, we distinguish between two cases. First, if item $g$ is allocated to agent $i$ in the last round by Round-Robin, we have $v_i(A_i)=v_i(A_i^{L-1}\cup \{g\})$ and, hence, $v_i(M)\leq v_i(A_i)$. Otherwise, monotonicity implies that $v_i(A_i^{L-1}\cup \{g\})\leq v_i(A_i\cup \{g\})$ and, hence, $v_i(M)\leq n\cdot v_i(A_i\cup \{g\})$, completing the proof.
\end{proof}

The result in Corollary~\ref{cor:rr} is best possible. The proof of the next theorem uses an allocation instance with monotone valuations that are slightly more general than submodular.

\begin{theorem}\label{thm:rr-not-prop1-xos}
    There exists an allocation instance with monotone XOS goods for which Round-Robin does not produce a PROP1 allocation.
\end{theorem}

\begin{proof}
We prove the theorem using an allocation instance with two agents and six items $a$, $b$, $c$, $d$, $e$, and $f$. Agent $1$ has the additive valuation function depicted in Table~\ref{tab:rr-not-prop1-xos-agent-1}.
\begin{table}[ht]
    \centering
    \begin{tabular}{c|c c c c c c}
        item & $a$ & $b$ & $c$ & $d$ & $e$ & $f$ \\\hline
        $v_1(\cdot)$ & $3$ & $2$ & $1$ & $0$ & $0$ & $0$      
    \end{tabular}
    \caption{The additive valuation function of agent $1$ used in the proof of Theorem~\ref{thm:rr-not-prop1-xos}.}
    \label{tab:rr-not-prop1-xos-agent-1}
\end{table}

Agent $2$ has an XOS valuation function $v_2$ that uses the two additive valuation functions $f_1$ and $f_2$ depicted in Table~\ref{tab:rr-not-prop1-xos-agent-2}. In particular, $v_2(S)=\max\{f_1(S), f_2(S)\}$ for every $S\subseteq \{a,b,c,d,e,f\}$.
\begin{table}[ht]
    \centering
    \begin{tabular}{c|c c c c c c}
        item & $a$ & $b$ & $c$ & $d$ & $e$ & $f$ \\\hline
        $f_1(\cdot)$ & $5$ & $5$ & $5$ & $0$ & $0$ & $4$\\
        $f_2(\cdot)$ & $0$ & $0$ & $0$ & $6$ & $2$ & $1$\\
    \end{tabular}
    \caption{The additive valuation functions that are used in the definition of the XOS valuation function $v_2$ of agent $2$ in the proof of Theorem~\ref{thm:rr-not-prop1-xos}.}
    \label{tab:rr-not-prop1-xos-agent-2}
\end{table}

Agent $2$ has item $d$ as the most valuable singleton. Then, the marginal valuation $v_2(o|d)$ is maximized for item $e$ with $v_2(e|d)=2$. Also, the marginal valuation $v_2(o|\{d,e\})$ is maximized to $1$ for item $f$; any other marginal is equal to $0$. So, in any execution of Round-Robin, agent $2$ will receive the items $d$, $e$, and $f$ (in this order) if they are available. Clearly, agent $1$ has the items $a$, $b$, and $c$ as the most valuable ones. 

Hence, in any of the two possible executions of Round-Robin (depending on the ordering of the agents), the resulting allocation will be $A=(\{a,b,c\},\{d,e,f\})$. Notice that $v_2(\{a,b,c,d,e,f\})=19$ while for every $o\in \{a,b,c\}$, it holds $v_2(\{d,e,f,o\})=9<v_2(\{a,b,c,d,e,f\})/2$, implying that allocation $A$ is not PROP1.
\end{proof}

\section{Computing PROP1 allocations for satiating subadditive goods}\label{sec:subadditive}

In this section, we present our strongest algorithmic result, stated as follows.
\begin{theorem}\label{thm:subadditive}
    There exists a polynomial-time algorithm that, on input any allocation instance with satiating subadditive goods, returns a PROP1 allocation.
\end{theorem}

\begin{proof}
We will prove Theorem~\ref{thm:subadditive} using Algorithm~\ref{alg:subadditive}. The input of the algorithm is an allocation instance with $n$ agents and a set $M$ of $m$ satiating subadditive goods. The algorithm uses the three subroutines $P(\cdot)$, $P1(\cdot)$, and $\im(\cdot)$. $P(\cdot)$ takes as input an allocation and returns the set of agents that satisfy the proportionality conditions. Formally, an agent $i$ belongs to set $P(A)$ for an allocation $A=(A_1, A_2, ..., A_n)$ if $n\cdot v_i(A_i)\geq v_i(M)$. Similarly, $P1(\cdot)$ takes as input an allocation and returns the set of agents that satisfy the PROP1 conditions. Formally, an agent $i$ belongs to set $P1(A)$ for an allocation $A=(A_1, A_2, ..., A_n)$ if it is either $n\cdot v_i(A_i) \geq v_i(M)$ (i.e., if $i\in P(A)$) or there exists an item $g\not\in A_i$ such that $n\cdot v_i(A_i\cup \{g\}) \geq v_i(M)$.

\begin{algorithm}[ht]
\caption{An algorithm producing a PROP1 allocation for allocation instances with satiating subadditive goods}\label{alg:subadditive}
\begin{algorithmic}[1]
\REQUIRE An allocation instance with $n$ agents and a set of satiating subadditive goods $M$
\ENSURE A PROP1 allocation $A$
\STATE $A\gets$ an arbitrary allocation
\STATE $A \gets \im(A)$
\WHILE {$P1(A)\not=[n]$}
\STATE $i\gets$ an arbitrary agent in $[n]\setminus P1(A)$
\STATE $j\gets$ an arbitrary agent in $P(A)$
\STATE $g\gets$ an arbitrary item in $A_j$ 
\STATE $A_j\gets A_j\setminus \{g\}$; $A_i\gets A_i\cup \{g\}$
\STATE $A\gets \im(A)$
\ENDWHILE
\RETURN $A$
\end{algorithmic}
\end{algorithm}

The subroutine $\im(\cdot)$ takes as input an allocation $A$ and works as follows. If there is an agent $i\in [n]$ with $i\in \argmax_{j\in [n]}{v_i(A_j)}$, $\im(A)$ returns $A$. Otherwise, it builds the directed graph $G(A)$ containing a node corresponding to each agent $i\in [n]$. For each agent $i\in [n]$, $G(A)$ contains the directed edge $(i,e_i)$ where $e_i$ is an agent in $[n]$ that agent $i$ envies the most, i.e., $e_i\in \argmax_{j\in [n]\setminus\{i\}}{v_i(A_j)}$. Since each node of graph $G(A)$ has out-degree $1$, $G(A)$ contains at least one cycle. The call of $\im(A)$ identifies an arbitrary such cycle $C=(c_1, ...c_{k})$ with $e_{c_k}=c_1$ and $e_{c_t}=c_{t+1}$ for $t=1, 2, ..., k-1$ and redistributes the bundles of $A$ so that the agent (corresponding to node) $c_k$ gets bundle $A_{c_1}$ and agent $c_t$ gets bundle $A_{c_{t+1}}$ for $t=1, 2, ..., k-1$. By the definition of subroutine $\im(\cdot)$, we have the following property.

\begin{lemma}\label{lem:after-improve}
Let $A$ be the allocation returned after an application of subroutine $\im(\cdot)$. Then, $P(A)\not=\emptyset$. Furthermore, every agent who was reassigned a bundle by subrouting $\im(\cdot)$ belongs to set $P(A)$.
\end{lemma}

\begin{proof}
Consider the application of subroutine $\im(\cdot)$ on an allocation $A$. If it did not modify allocation $A$, this means that there is some agent $i\in [n]$ such that $v_i(A_i)\geq v_i(A_j)$ for every agent $j\in [n]$. Summing over these $n$ inequalities and using the subadditivity of the valuation function $v_i$, we get
    \begin{align*}
        n\cdot v_i(A_i) &\geq \sum_{j\in [n]}{v_i(A_j)} \geq v_i(M),
    \end{align*}
    implying the proportionality condition for agent $i$ and, hence, $i\in P(A)$.

If the application of subroutine $\im(\cdot)$ modified allocation $A$, then a non-empty set of agents were reallocated their most valuable bundle. For each such agent $i$, their valuation after the redistribution of the bundles is $v_i(A_i)\geq v_i(A_j)$ for every $j\in [n]$ (here, we use $A$ to denote the allocation obtained after the application of subrouting $\im(\cdot)$). Using the same argument as in the previous paragraph, we get that all these agents satisfy the proportionality conditions and belong to $P(A)$.
\end{proof}

Algorithm~\ref{alg:subadditive} works as follows. It starts with an arbitrary allocation (line 1) on which subrouting $\im(\cdot)$ is applied (line 2). Then, as long as there are agents for whom the PROP1 condition is not satisfied (i.e., the condition $P1(A)\not=[n]$ in line 3), the algorithm runs the following loop. It identifies an arbitrary agent $i$ for whom the PROP1 condition is not satisfied (line 4) and an arbitrary agent $j$ for whom the proportionality condition is satisfied (line 5). Such agents do exist since the algorithm entered the while-loop and allocation $A$ has been obtained after applying subroutine $\im(\cdot)$. Then, the algorithm picks an arbitrary item $g$ from the bundle $A_j$ (in line 6; again, such an item exists because agent $j$ belongs to $P(A)$ and, thus, satisfies the proportionality condition, i.e., $v_j(A_j)\geq v_j(M)/n>0$). The current allocation is modified (in line 7) by moving item $g$ from bundle $A_j$ to bundle $A_i$, and subroutine $\im(\cdot)$ is applied to this modified allocation (line 8). When leaving the while-loop, the algorithm returns the current allocation (line 10).

The above discussion guarantees that all steps in the while-loop are well-defined. Furthermore, if the algorithm exits the while-loop and terminates, it will clearly return a PROP1 allocation (since $P1(A)=[n]$, i.e., all agents would satisfy the PROP1 condition in this case). To complete the proof of correctness, we need to prove that the algorithm terminates given any allocation instance with satiating subadditive goods.

We use a potential function argument. Denote by $A^0$ the allocation obtained after the execution of line 2 of Algorithm~\ref{alg:subadditive} and by $A^t$ the allocation obtained after the $t$-th execution of line 8 (i.e., after the $t$-th execution of the while-loop). With some abuse in notation, we use $P1(A^t)$ and $P(A^t)$ to denote the set of agents for whom the PROP1 and proportionality conditions are satisfied in allocation $A^t$. 

\begin{lemma}\label{lem:subset}
    For every integer $t\geq 0$, $P1(A^t)\subseteq P1(A^{t+1})$.
\end{lemma}

\begin{proof}
Let $i$, $j$, and $g$ be the agent of $[n]\setminus P1(A^t)$, the agent of $P(A^t)$, and the item of bundle $A^t_j$ selected during the $(t+1)$-th execution of the while-loop, respectively. Denote by $A'$ the current allocation after the execution of line 7 (and before the execution of line 8) in the $(t+1)$-th execution of the while-loop. Notice that the set $P1(A^t)\setminus\{j\}$ consist of agents $k$ satisfying $A'_k=A^t_k$, since the only difference between allocations $A^t$ and $A'$ is the item $g$ which is moved from the bundle of agent $j$ to that of agent $i$ who, by definition, does not belong to set $P1(A^t)$. For agent $j$, if her valuation increased after removing item $g$ from her bundle, then
\begin{align*}
    v_j(A'_j) &> v_j(A_j^t)\geq v_j(M)/n.
\end{align*}
Otherwise,
\begin{align*}
    \max_{g'\in M\setminus A'_j}{v_j(A'_j\cup \{g'\})} &\geq v_j(A'_j\cup \{g\})\\
    &= v_j(A^t_j) \geq v_j(M)/n.
\end{align*}
The last inequality in the two derivations above is due to the fact $j\in P(A^t)$. Hence, in any case, agent $j$ belongs to $P1(A')$ and, thus, $P1(A^t)\subseteq P1(A')$. 

The allocation $A^{t+1}$ is obtained by applying subroutine $\im(\cdot)$ on allocation $A'$. For every agent of $P1(A')$ who gets the same bundle in $A'$ and $A^{t+1}$, the PROP1 conditions are trivially satisfied in $A^{t+1}$ and the agent belongs to $P1(A^{t+1})$. By Lemma~\ref{lem:after-improve}, every agent of $P1(A')$ who was reassigned a new bundle during the application of subroutine $\im(\cdot)$
on allocation $A'$, satisfies the proportionality condition and belongs to $P(A^{t+1})$ and, consequently, to $P1(A^{t+1})$. Hence, $P1(A')\subseteq P1(A^{t+1})$, which completes the proof.
\end{proof}
Lemma~\ref{lem:subset} proves the monotonicity of the quantity $|P1(A^t)|$ in terms of number of executions $t$ of the while-loop of Algorithm~\ref{alg:subadditive}. To show that the algorithm terminates, we furthermore need to show that $|P1(A^t)|$ strictly increases after a polynomial number of executions of the while-loop. This is proved in the next lemma.
\begin{lemma}\label{lem:pot-increases}
    For every non-negative integer $t$ with $|P1(A^t)|<n$, there is an integer $t'\leq t+m$ such that $|P1(A^{t'})|>|P1(A^t)|$.
\end{lemma}

\begin{proof}
For the sake of contradiction, let $t$ be a non-negative integer such that $|P1(A^t)|<n$ and assume that $|P1(A^{t+m})|\leq |P1(A^t)|$. By Lemma~\ref{lem:subset}, this means that $P1(A^{t+m})=P1(A^t)$. Hence, after the $t$-th execution and until after the $(t+m)$-th execution of the while-loop, $m$ distinct items (i.e., all items) have moved to bundles of agents in $[n]\setminus P1(A^{t+m})$. The agents in set $P1(A^{t+m})$ have no value at all in allocation $A^{t+m}$. Hence, $v_k(A^{t+m}_k)<v_k(M)/n$, for each agent $k\in [n]$ and $P(A^{t+m})=\emptyset$, which contradicts Lemma~\ref{lem:after-improve}, as allocation $A^{t+m}$ is obtained after running the subroutine $\im(\cdot)$. 
\end{proof}

Hence, Lemma~\ref{lem:pot-increases} implies that every at most $m$ executions of the while-loop, the number of agent satisfying the PROP1 conditions increases by at least $1$. This means that Algorithm~\ref{alg:subadditive} computes a PROP1 allocation after at most $n\cdot m$ executions of the while-loop. The proof of Theorem~\ref{thm:subadditive} is now complete.
\end{proof}

The proof of Theorem~\ref{thm:subadditive} essentially shows that Algorithm~\ref{alg:subadditive} computes a PROP1 allocation for satiating subadditive goods after $O(n\cdot m)$ item allocations or reallocations. For allocation instances with satiating submodular goods, we can compute a PROP1 allocation much faster by combining Round-Robin and Algorithm~\ref{alg:subadditive}. The main change is to replace line 1 of Algorithm~\ref{alg:subadditive} by an execution of Round-Robin until its second-to-last round. By Theorem~\ref{thm:rr-second-last-phase}, the partial allocation computed by Round-Robin after the second-to-last round satisfies PROP1 for each agent. We can complete this allocation by assigning all items left for the last Round-Robin round to a single arbitrary agent; this complete allocation satisfies PROP1 for all agents but one. By running Algorithm~\ref{alg:subadditive} using this initial allocation, we have $|P1(A^0)|\geq n-1$ initially, and, using Lemma~\ref{lem:pot-increases}, we get that all agents will satisfy PROP1 after at most $m$ executions of the while-loop of Algorithm~\ref{alg:subadditive}. The following statement summarizes this discussion.

\begin{theorem}
    There exists a polynomial-time algorithm which, given an allocation instance with $m$ satiating submodular goods, returns a PROP1 allocation after at most $O(m)$ item allocations and reallocations.
\end{theorem}

\section{Can PROP1 allocations be Pareto-optimal?}\label{sec:prop1-and-pareto}
We now explore the compatibility of PROP1 and Pareto-optimality. Following~\citet{CKM+19}, we extend the definition of the maximum Nash welfare allocation to capture instances in which no allocation yields positive valuation to all agents for their bundles. So, the maximum Nash welfare allocation is one that maximizes the number of agents with non-zero valuation for their bundle and, under this condition, it maximizes the product of non-zero agent valuations.

Our first result showcases another facet of the ``unreasonable fairness'' of allocations with maximum Nash welfare.

\begin{theorem}\label{thm:max-nsw}
    In any allocation instance with monotone submodular goods, the maximum Nash welfare allocation is PROP1.
\end{theorem}

We give two proofs of Theorem~\ref{thm:max-nsw}. The first one is direct and long. A shorter proof follows by first using a result of~\citet{CKM+19} stating that maximum Nash welfare allocations have a fairness property called {\em marginal envy-freeness up to some good} (MEF1) and then proving that MEF1 allocations are also PROP1. 

\paragraph{A direct proof of Theorem~\ref{thm:max-nsw}.}
Consider an allocation instance with $n$ agents having monotone submodular valuations over goods, and let $A=(A_1,A_2, ..., A_n)$ be a maximum Nash welfare allocation. Let $i\in [n]$ be an agent and $g^*\in \argmax_{g\in M\setminus A_i}{v_i(g|A_i)}$. We will show that $v_i(M)\leq n\cdot v_i(A_i\cup\{g^*\})$.

Define $\widetilde{N}:=\{j\in [n]: v_j(A_j)>0\}$. We distinguish between two cases depending on whether $i\in \widetilde{N}$ and $i\in [n]\setminus \widetilde{N}$.

\paragraph{Case 1: $i\in \widetilde{N}$.}
Let $j$ be an agent in $\widetilde{N}$ different from $i$ and $g$ an item in bundle $A_j$. Since $A$ is a maximum Nash welfare allocation, the product of the valuations the agents in $\widetilde{N}$ have for their bundles is at least as high as that in the allocation obtained by $A$ after removing item $g$ from bundle $A_j$ and putting it into bundle $A_i$. Thus,
\begin{align*}
    0 &\geq v_i(A_i\cup\{g\})\cdot v_j(A_j\setminus\{g\})-v_i(A_i)\cdot v_j(A_j)\\ &=\left(v_i(A_i)+v_i(g|A_i)\right)\cdot \left(v_j(A_j)-v_j(g|A_j\setminus\{g\})\right)\\
    &\quad\, -v_i(A_i)\cdot v_j(A_j)\\
    &=-(v_i(A_i)+v_i(g|A_i))\cdot v_j(g|A_j\setminus \{g\})\\
    &\quad\, +v_i(g|A_i)\cdot v_j(A_j)\\
    &=-v_i(A_i\cup\{g\})\cdot v_j(g|A_j\setminus\{g\})+v_i(g|A_i)\cdot v_j(A_j)
\end{align*}
and, equivalently,
\begin{align}\label{eq:mnw-consequence-non-zero}
    v_i(g|A_i) &\leq \frac{v_i(A_i\cup \{g\})\cdot v_j(g|A_j\setminus\{g\})}{v_j(A_j)}.
\end{align}

Now, let $j$ be an agent in $[n]\setminus \widetilde{N}$ (if any) and $g$ an item in bundle $A_j$. It holds that 
\begin{align}\label{eq:mnw-consequence-zero}
    v_i(g|A_i) &= 0,
\end{align}
since otherwise it would be $v_i(A_i\cup \{g\})=v_i(A_i)+v_i(g|A_i)>v_i(A_i)$ and we could move item $g$ from bundle $A_j$ to bundle $A_i$ to get an even higher product of valuations of the agents in $\widetilde{N}$ for their bundles, contradicting the definition of allocation $A$.

Then, using Claim~\ref{claim:submodular-vs-additive}, Equations (\ref{eq:mnw-consequence-non-zero}) and~(\ref{eq:mnw-consequence-zero}), and the definition of item $g^*$, we get
\begin{align}\nonumber
    &v_i(M\setminus A_i|A_i)\\\nonumber
    &\leq \sum_{g\in M\setminus A_i}{v_i(g|A_i)}\\\nonumber
    &= \sum_{j\in \widetilde{N}\setminus\{i\}}{\sum_{g\in A_j}{v_i(g|A_i)}}\\\nonumber
    &\leq \sum_{j\in \widetilde{N}\setminus\{i\}}{\sum_{g\in A_j}{\frac{v_i(A_i\cup \{g\})\cdot v_j(g|A_j\setminus\{g\})}{v_j(A_j)}}}\\\nonumber
    &\leq v_i(A_i\cup\{g^*\})\\\label{eq:sum-of-last-marginals}
    &\quad\quad\,\cdot \sum_{j\in \widetilde{N}\setminus\{i\}}{\frac{1}{v_j(A_j)}\cdot \sum_{g\in A_j}{v_j(g|A_j\setminus\{g\})}}.
\end{align}
Now, for each agent $j\in \widetilde{N}\setminus \{i\}$, rename as $g_1$, $g_2$, ..., $g_{|A_j|}$ the items in bundle $A_j$ and, for $t=1, 2, ..., |A_j|-1$, let $A_j^t=\{g_k|1\leq k \leq t\}$ be the sub-bundle of $A_j$ consisting of the items $g_1$, $g_2$, ..., $g_{t-1}$ (if any). Then, $A_j^t\subseteq A_j\setminus \{g_t\}$ and, using the submodularity of $v_j$, we obtain
\begin{align}\nonumber
\sum_{g\in A_j}{v_j(g|A_j\setminus \{g\})} &= \sum_{t=1}^{|A_j|}{v_j(g_t|A_j\setminus \{g_t\})}\\\label{eq:sum-of-last-marginals-bound}
&\leq \sum_{t=1}^{|A_j|}{v_j(g_t|A_j^{t-1})} = v_j(A_j).
\end{align}
Using Equation (\ref{eq:sum-of-last-marginals-bound}), Equation (\ref{eq:sum-of-last-marginals}) now yields
\begin{align}\label{eq:M-setminus_A_i-given-A_i} 
v_i(M\setminus A_i|A_i)&\leq (|\widetilde{N}|-1)\cdot v_i(A_i\cup\{g^*\}).
\end{align}
Hence, using Equation (\ref{eq:M-setminus_A_i-given-A_i}), monotonicity, and the fact $|\widetilde{N}|\leq n$, we obtain
\begin{align*}
    v_i(M) &=v_i(A_i)+v_i(M\setminus A_i|A_i)\\
    &\leq n\cdot v_i(A_i\cup\{g^*\}),
\end{align*}
as desired.

\paragraph{Case 2: $i\in [n]\setminus \widetilde{N}$.} Define the set of items $G_i:=\{g\in M\setminus A_i: v_i(g|A_i)>0\}$. We will prove that $|G_i|\leq n-1$. The proof will then follow since, using Claim~\ref{claim:submodular-vs-additive}, the definition of set $G_i$, and monotonicity, we can get 
    \begin{align*}
        v_i(M) &=v_i(A_i)+v_i(M\setminus A_i|A_i)\\
        & \leq v_i(A_i)+\sum_{g\in M\setminus A_i}{v_i(g|A_i)}\\
        &= v_i(A_i)+\sum_{g\in G_i}{v_i(g|A_i)}\\
        &\leq v_i(A_i)+(n-1)\cdot v_i(g^*|A_i)\\
        &\leq n\cdot v_i(A_i\cup\{g^*\}),
    \end{align*}
    i.e., the PROP1 condition is satisfied for agent $i$.

It remains to show that $|G_i|\leq n-1$. First notice that $v_i(g|A_i)=0$ for every item $g$ allocated to an agent in $[n]\setminus \widetilde{N}$ different than $i$. Indeed, if $v_i(g|A_i)>0$ for an item $g$ allocated to such an agent $i'$, then moving item $g$ from bundle $A_{i'}$ to bundle $A_i$, we would obtain an allocation with a strictly larger number of agents with positive valuation for their bundles compared to $A$, contradicting its definition. 
    
Now, assume that $|G_i|\geq n$. Since all items in $G_i$ have been allocated to agents in $\widetilde{N}$ under $A$, there is an agent $j\in \widetilde{N}$ and two items $a$ and $b$ in $A_j$ such that $v_i(a|A_i)>0$ and $v_i(b|A_i)>0$. We now claim that $v_j(A_j\setminus \{a\})=v_j(A_j\setminus \{b\})=0$. Indeed, if, e.g., $v_j(A_j\setminus \{a\})>0$, then moving item $a$ from the bundle $A_j$ to the bundle $A_i$, we would get an allocation with a strictly larger number of agents with positive value compared to $A$, again contradicting its definition. 
    
We now show that the conditions $v_j(A_j\setminus \{a\})=v_j(A_j\setminus \{b\})=0$ and $v_j(A_j)>0$ violate submodularity. Indeed, we have $0=v_j(A_j\setminus \{a\})=v_j(A_j\setminus \{a,b\})+v_j(b|A_j\setminus \{a,b\})$, which implies that $v_j(A_j\setminus \{a,b\})=0$ and $v_j(b|A_j\setminus \{a,b\})=0$. We also have $0<v_j(A_j)=v_j(A_j\setminus \{b\})+v_j(b|A_j\setminus \{b\})=v_j(b|A_j\setminus \{b\})$. Hence, $v_j(b|A_j\setminus \{b\})>v_j(b|A_j\setminus \{a,b\})$, contradicting the submodularity of valuation function $v_j$. Our direct proof of the claim $|G_i|\leq n-1$ is now complete.
\qed

\paragraph{An alternative proof of Theorem~\ref{thm:max-nsw}.}
Our alternative proof of Theorem~\ref{thm:max-nsw} uses a result of~\citet{CKM+19}. Specifically,~\citet[Theorem 3.5]{CKM+19} proved that in allocation instances with monotone submodular goods, the maximum Nash welfare allocation satisfies the fairness criterion called {\em marginal envy-freeness up to some good} (MEF1), which is defined as follows. 

\begin{definition}
    An allocation $A=(A_1,A_2, ..., A_n)$ is marginal envy-free up to some good (MEF1) in an allocation instance with $n$ agents and monotone submodular goods, if for every two agents $i$ and $j$ with $A_j\not=\emptyset$, there exists a good $g\in A_j$ such that $v_i(A_i)\geq v_i(A_j\setminus \{g\}|A_i)$.
\end{definition}

Our alternative proof of Theorem~\ref{thm:max-nsw} follows after showing (via Lemma~\ref{lem:mef1-implies-prop1}) that MEF1 implies PROP1.
\begin{lemma}\label{lem:mef1-implies-prop1}
    In any allocation instance with monotone submodular goods, any MEF1 allocation is also PROP1.
\end{lemma}

\begin{proof}
Consider an allocation instance with $n$ agents having monotone submodular valuations over goods and let $A=(A_1, A_2, ..., A_n)$ be an MEF1 allocation. We will prove that there exists a good $g\not\in A_1$ such that $v_1(M)\leq n\cdot v_1(A_1\cup\{g\})$. This will prove that the allocation $A$ satisfies the PROP1 condition for agent $1$; the proof of the PROP1 condition for any other agent trivially follows by applying the same argument after renaming the agents appropriately.

For agent $j=2, ..., n$ such that $A_j\not=\emptyset$, denote by $g_j$ a good in bundle $A_j$ such that
\begin{align*}
    v_1(A_1) &\geq v_1(A_j\setminus\{g_j\}|A_1).
\end{align*}
By the MEF1 property of allocation $A$, such a good certainly exists. Let $G$ be the union of all items $g_j$ for agent $j\in [n]\setminus\{1\}$ such that $A_j\not=\emptyset$, i.e., $G:=\{g_j\in A_j:j\in [n]\setminus \{1\}, A_j\not=\emptyset\}$ Also, let $g^*\in \argmax_{g\in G}{v_1(g|A_1)}$.

For agent $j\in [n]\setminus \{1\}$ such that $A_j\not=\emptyset$, we have
\begin{align}\nonumber
    v_1(A_1\cup\{g^*\}) &\geq v_1(A_1\cup\{g_j\}) = v_1(A_1)+v_1(g_j|A_1)\\\nonumber
    &\geq v_1(A_j\setminus \{g_j\}|A_1)+v_1(g_j|A_1)\\\label{eq:marginal}
    &\geq v_1(A_j|A_1)\geq v_1\left(A_j|\cup_{t=1}^{j-1}{A_t})\right).
\end{align}
The first inequality follows from the definition of good $g^*$, the second one from the MEF1 condition, the third one by Claim~\ref{claim:submodular-vs-additive}, and the fourth one by the submodularity of function $v_1$.

Notice that inequality (\ref{eq:marginal}) trivially holds if $A_j=\emptyset$ and for $j=1$ (due to monotonicity). Thus, summing inequality (\ref{eq:marginal}) for $j\in [n]$, we get
\begin{align*}
    n\cdot v_1(A_1\cup\{g^*\}) &\geq \sum_{j=1}^n{v_1\left(A_j|\cup_{t=1}^{j-1}{A_t}\right)}=v_1(M),
\end{align*}
which implies the PROP1 condition for agent $1$.
\end{proof}

As a maximum Nash welfare allocation is Pareto-optimal, we obtain the following corollary.

\begin{corollary}
In any allocation instance with monotone submodular goods, a PROP1 and Pareto-optimal allocation exists.
\end{corollary}

Unfortunately, PROP1 can be incompatible with Pareto-optimality if slightly more general (i.e., monotone XOS) valuation functions are allowed.

\begin{theorem}\label{thm:xos-po-and-prop1-incompatible}
There exists an allocation instance with monotone XOS goods, in which no Pareto-optimal allocation is PROP1.
\end{theorem}

\begin{proof}
We use the allocation instance with three agents having the same XOS valuation function for seven items that we used in the proof of Theorem~\ref{thm:ef1-not-imply-prop1}. The only PROP1 allocations should give three items to some agent and two items to each of the other two. Indeed, if an agent got at most one item, the valuation for any superset including an extra item would be $2$ while the proportionality threshold is $7/3$. Then, the (non-PROP1) allocation that gives five items to the agent who initially got the three items and one item to each of the other two agents is a Pareto improvement; the first agent has strictly better value and the other two are not worse off.
\end{proof}

\section{Conclusion and open problems}\label{sec:open}
We have presented new results regarding the existence and efficient computation of PROP1 allocations for instances with non-additive (i.e., submodular and subadditive) and possibly non-monotone valuation functions over goods. Our work leaves several open problems, including the following. First, is PROP1 and Pareto-optimality compatible on allocation instances with satiating submodular goods? Second, is there a polynomial-time algorithm for computing a PROP1 and Pareto-optimal allocation for monotone submodular goods? Third, what is the price of fairness with respect to Nash welfare of PROP1 allocations in subadditive instances? The proof of Theorem~\ref{thm:xos-po-and-prop1-incompatible} can be modified to yield a price of fairness lower bound higher than $1$. Finally, for classes of allocation instances in which PROP1 is not compatible with Pareto-optimality, what is the best approximation of PROP1 that is compatible with Pareto-optimality? Can such approximately PROP1 and Pareto-optimal allocations be computed efficiently?

\section*{Acknowledgements}
This work has been partially supported by Independent Research Fund Denmark (DFF) under grant 2032-00185B. 

\bibliography{sample}

\end{document}